\documentclass[british,orivec]{llncs}

\usepackage[utf8]{inputenc}

\usepackage{amsmath}
\usepackage{amsfonts}
\usepackage{amssymb}
\usepackage{stmaryrd}
\SetSymbolFont{stmry}{bold}{U}{stmry}{m}{n}
\usepackage{bm}
\usepackage{mathtools}
\usepackage{centernot}
\usepackage{scalerel}
\usepackage{lscape}
\usepackage{mathrsfs}

\usepackage{hyperref}
\usepackage[capitalise]{cleveref}
\allowdisplaybreaks

\makeatletter
\renewcommand{\xRightarrow}[2][]{\ext@arrow 0359\Rightarrowfill@{#1}{#2}}
\makeatother

\usepackage{array}

\usepackage[inline]{enumitem}

\usepackage{graphicx}
\usepackage{tikz}
\usetikzlibrary{arrows,positioning,arrows.meta,decorations.pathmorphing}
\tikzset{
  squiggly/.style = {
    line join=round,
    decorate, decoration={
      zigzag,
      segment length=4,
      amplitude=.9,post=lineto,
      post length=2pt}
  }
}

\usepackage{breakurl}

\usepackage{proof-dashed}

\newcounter{rule}

\crefname{infrule}{rule}{rules}
\Crefname{infrule}{Rule}{Rules}

\usepackage{cite}

\def\mod{\mathbin{\mathrm{mod}}}

\newcommand{\m}[1]{\ensuremath{\mathsf{#1}}}
\newcommand{\til}{\vec}
\newcommand{\pid}[1]{\m{#1}}
\newcommand{\pids}[1]{\ensuremath{\til{\pid #1}}}
\newcommand{\nil}{\ensuremath{\boldsymbol{0}}}
\newcommand{\com}[4]{\ensuremath{\pid{#1}.{#2} \tto\pid{#3}.{#4}}}
\newcommand{\lbl}[1]{\textsc{#1}}
\newcommand{\albl}{\lbl{l}}
\newcommand{\sel}[3]{\ensuremath{\pid{#1}\tto\pid{#2}[#3]}}
\newcommand{\gencom}{\ensuremath{\com peqx}}
\newcommand{\gensel}{\ensuremath{\sel pq\albl}}
\newcommand{\cond}[4]{\condif \pid{#1}.{#2} \condthen #3 \condelse #4}
\newcommand{\condif}{\mathop{\m{if}}}
\newcommand{\condthen}{\mathop{\m{then}}}
\newcommand{\condelse}{\mathop{\m{else}}}
\newcommand{\gencond}{\ensuremath{\cond pb{C_1}{C_2}}}
\newcommand{\cont}[3]{{\ensuremath{\lceil{#2},{#1}\rfloor{#3}}}}
\newcommand{\gencont}{\ensuremath{\cont X{\pids q}C}}
\newcommand{\passign}[2]{#1 \mathbin{\coloneqq} #2}
\newcommand{\assign}[3]{\pid{#1}.\passign{#2}{#3}}
\newcommand{\cassign}[3]{#1.\passign{#2}{#3}}
\newcommand{\gencassign}{\ensuremath{\cassign{\pid p}xe}}
\newcommand{\cdefs}{\ensuremath{\mathscr{C}}}
\newcommand{\tuple}[1]{\ensuremath{\langle #1 \rangle}}
\newcommand{\lto}[1]{\ltoc{#1}\cdefs}
\newcommand{\ltoc}[2]{\xrightarrow{#1}_{#2}}
\newcommand{\mlto}{\ensuremath{\to_\cdefs^\ast}}
\newcommand{\rname}[2]{{\small\textsc{#1}\ensuremath{|}\textsc{#2}}}
\newcommand{\tto}{\ensuremath{\mathbin{\boldsymbol{\rightarrow}}}}
\newcommand{\true}{\m{true}}
\newcommand{\false}{\m{false}}
\newcommand{\pn}{\mathop{\mathsf{pn}}}
\newcommand{\coml}[3]{\ensuremath{\pid{#1}.{#2} \tto\pid{#3}}}
\newcommand{\gencoml}{\ensuremath{\coml pvq}}
\newcommand{\condl}[1]{\tau@{\pid{#1}}}
\newcommand{\gencondl}{\ensuremath{\condl p}}
\newcommand{\allpids}{\ensuremath{\mathcal P}}
\newcommand{\cprops}{\ensuremath{\mathfrak{C}}}
\newcommand{\hto}[2][\cdefs]{\mathrel{\xRightarrow{#2}_{#1}}}
\newcommand{\mhto}{\ensuremath{\Rightarrow_\cdefs^\ast}}
\newcommand{\wlp}[1][\cprops]{\ensuremath{\mathsf{wlp}_{#1}}}
\newcommand{\logeq}[1]{\mathrel{\stackrel{\logvar{#1}}=}}
\newcommand{\disjoint}{\mathbin{\#}}
\newcommand{\hoare}[3]{\ensuremath{\{#1\}{#2}\{#3\}}}
\newcommand{\logvar}[1]{\ensuremath{\mathcal #1}}
\newcommand{\decth}{\ensuremath{\mathfrak D}}
\newcommand{\subst}[5]{\ensuremath{{#1}[\pid{#2}.{#3}:=\pid{#4}.{#5}]}}
\newcommand{\loc}[2]{\ensuremath{L(\pid{#1},{#2})}}
\newcommand{\eval}[4]{\ensuremath{{#1}\downarrow_{{#2}(\pid{#3})}{#4}}}
\newcommand{\rulebreak}{\\[1em]}
\newcommand{\eoe}{\hspace*\fill\ensuremath{\triangleleft}}

\usepackage{orcidlink}
\def\orcidID#1{\unskip$^{\orcidlink{#1}}$}

\begin{document}
\title{Reasoning about Choreographic Programs}

\author{Luís Cruz-Filipe \orcidID{0000-0002-7866-7484} \and Eva Graversen \orcidID{0000-0002-9430-4907} \and Fabrizio Montesi \orcidID{0000-0003-4666-901X} \and Marco Peressotti \orcidID{0000-0002-0243-0480}}

\institute{Department of Mathematics and Computer Science, University of Southern Denmark}%

\authorrunning{L.~Cruz-Filipe et al.}

\maketitle
\begin{abstract}
  Choreographic programming is a paradigm where a concurrent or distributed system is developed in a top-down fashion.
  Programs, called choreographies, detail the desired interactions between processes, and can be compiled to distributed implementations based on message passing.
  Choreographic languages usually guarantee deadlock-freedom and provide an operational correspondence between choreographies and their compiled implementations, but until now little work has been done on verifying other properties.

  This paper presents a Hoare-style logic for reasoning about the behaviour of choreographies, and illustrate its usage in representative examples.
  We show that this logic is sound and complete, and discuss decidability of its judgements.
  Using existing results from choreographic programming, we show that any functional correctness property proven for a choreography also holds for its compiled implementation.
\end{abstract}

\section{Introduction}
Programming communicating systems is hard, because of the challenge of ensuring that separate communication actions (like sending or receiving a message) executed by independent programs match each other correctly at runtime~\cite{LLLG16}.

In the paradigm of \emph{choreographic programming}~\cite{M13p}, this challenge is tackled by providing high-level abstractions that allow programmers to express the desired flow of communications safely from a `global' viewpoint~\cite{CM13,CGLMP22,CM17,DGGLM17,GMPRSW21,HG22,JV22,LNN16,M23}.
In a choreography program, or \emph{choreography}, communication is expressed in some variation of the communication term from security protocol notation, $\pid{Alice} \mathbin{\texttt{->}} \pid{Bob}\colon M$, which reads ``$\pid{Alice}$ communicates the message $M$ to $\pid{Bob}$''~\cite{NS78}.
These terms can be composed in structured choreographies using common programming language constructs.
Then, a compiler can automatically generate an executable distributed implementation~\cite{CM13,DGGLM17,GMP20}, as depicted in \cref{fig:chorintro}.

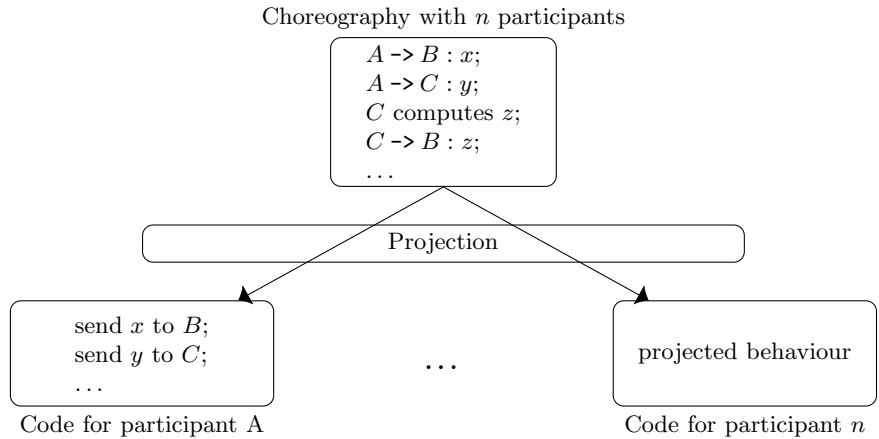
\begin{figure}
  \centering
  \begin{tikzpicture}[font=\footnotesize]
  \node [
    rectangle, rounded corners, draw,
    minimum width=3cm,
    label=above:Choreography with $n$ participants,
    align=left,
  ] (chor) {$A \mathbin{\texttt{->}} B:x;$ \\ $A \mathbin{\texttt{->}} C:y;$ \\ $C$ computes $z$; \\ $C \mathbin{\texttt{->}} B:z;$ \\ \dots};
  
  \node [
    rectangle, rounded corners, draw,
    below=0.5cm of chor,
    minimum width=8cm,
  ] (proj) {Projection};
  
  \node [
    rectangle, rounded corners, draw,
    below=0.5cm of proj.south west,
    minimum width=3.5cm,
    minimum height=1.4cm,
    label=below:Code for participant A,
    align=left,
  ] (A) {send $x$ to $B$;\\ send $y$ to $C$; \\ \dots};
  
  \node [
    below=1.25cm of proj,
  ] (dot) {\textbf{\dots }};
  
  \node [
    rectangle, rounded corners, draw,
    below=0.5cm of proj.south east,
    minimum width=3.5cm,
    minimum height=1.4cm,
    label=below:Code for participant $n$,
  ] (N) {projected behaviour};
  
  \draw[-{Latex[length=1.5mm,width=3mm]}] (chor.south) to (A);
  \draw[-{Latex[length=1.5mm,width=3mm]}] (chor.south) to (N);
  
  \end{tikzpicture}
  \caption{Choreographic programming: the communication and computation behaviour of a system is defined in a choreography, which is then projected (compiled) to deadlock-free distributed code (adapted from \cite{GMPRSW21}).}\label{fig:chorintro}
\end{figure}

So far, research on choreographic programming has mostly focused on improving the expressivity of choreographic programming languages, their implementation, and the formalisation of general properties about compilation.
Theory of choreographic programming typically comes with proofs of correctness of the accompanying compilation procedure.
A hallmark result is \emph{deadlock-freedom by design}: since mismatched communication actions cannot be syntactically expressed in choreographies, the compiled code cannot incur deadlocks~\cite{CM13}.

By contrast, little research has been done on general methods for proving functional correctness properties about choreographies.
Yet choreographies codify distributed protocols, and reasoning about the effect that these protocols have on the states of participants is usually important.

\paragraph*{This work.}
In this work, we present a Hoare logic for reasoning about choreographies.
Hoare logic~\cite{AO19,H69} is a common way of reasoning about programs.
A Hoare assertion is a triple, $\{\varphi\}P\{\psi\}$, where $\varphi$ and $\psi$ are formulas (respectively called the \emph{precondition} and \emph{postcondition}) and $P$ is a program.
This triple states that if $P$ is executed from a state that satisfies $\varphi$ and terminates, then the final state satisfies $\psi$.
We develop a Hoare logic where programs are choreographies and formulas can talk about the states of multiple processes jointly.

Our framework is based on well-studied theories of choreographic programming~\cite{CM20,M23}, in particular on properties that have been formalised in Coq~\cite{CMP21b,CMP21}.
This helps with the generality and elegance of our development.
For example, we leverage the property of confluence in metatheoretical proofs, and we rely on the compiler correctness results proven previously to transfer properties proven with our logic to distributed implementations compiled from choreographies.

\paragraph*{Contribution.}
We define a Hoare logic for reasoning about choreographic programs expressed in standard ways, thanks to a modular design parametrised on the language of state formulas.
We prove that our logic has the expected properties of a Hoare logic (soundness and partial completeness), and illustrate how it can be used to prove important properties of specific protocols encoded as choreographies.

\paragraph*{Structure.}
We review the choreographic language from \cite{CM20} in \cref{sec:lang}.
In \cref{sec:hoare} we describe our logic and prove its soundness.
\Cref{sec:compl} introduces weakest liberal preconditions, and uses them to show completeness and decidability results.
\cref{sec:related} discusses additional related work.
Illustrative examples are included throughout the text.

\section{Language}\label{sec:lang}
In this section we recall the choreographic language from~\cite{CM20}, which we will be reasoning about.
This language models systems of independent processes (networks), which interact by means of synchronous communication.
Each process is uniquely identified by a name, which is known by all other processes in the network, and can store values locally in memory referenced by variables.
The set of variable names is assumed to be the same for all processes.
The set of all processes is denoted by \allpids.

There are two kinds of messages that can be exchanged: \emph{values} are results of evaluating \emph{expressions} locally; and \emph{selection labels} are special constants used to implement agreement on choices about alternative distributed behaviour.

The actual sets of expressions and labels are left unspecified, but we make some assumptions.
\emph{Labels} are taken from a (small) finite set.
\emph{Expressions} are freely generated from a (typed) signature $\Xi$ and the set of process variables.
Expressions that evaluate to a Boolean value are also called Boolean expressions.

\subsection{Syntax}
Formally, the syntax of choreographies is defined by the grammar
\begin{align*}
C &::= I;C \mid \gencond \mid X \mid \gencont \mid \nil\\
I &::= \gencassign \mid \gencom \mid \gensel
\end{align*}
where $C$ is a choreography, $I$ is an instruction, $\pid p$ and $\pid q$ are processes names, $e$ is an expression, $v$ is a value, $x$ is a variable, $b$ is a Boolean expression, $\albl$ is a selection label, and $X$ is a procedure name.

Choreographies can be built as: an instruction $I$ followed by a choreography; alternative composition of two choreographies $C_1$ and $C_2$; procedure calls; or the terminated choreography $\nil$.
There are two terms for procedure calls, corresponding to: (a)~a procedure that has yet to be entered by any processes ($X$) or (b) one which has already started, annotated with the set of processes that still have to enter it ($\gencont$).

There are three types of instructions: local assignment (\gencassign), where $\pid p$ evaluates expression $e$ and stores the result in its local variable $x$; value communication, where $\pid p$ evaluates $e$ and sends the result to $\pid q$, who stores it in variable $x$; and label selection, where $\pid p$ sends a label $\albl$ to $\pid q$ (typically to communicate the result of a local choice -- see below).

In a conditional, $\gencond$, process $\pid p$ evaluates the expression $b$ to decide whether the choreography should continue as $C_1$ or $C_2$.
Since only $\pid p$ knows the result of the evaluation, the remaining processes need to be informed of how they should behave -- this knowledge is typically propagated to other participants by means of label selections.\footnote{For this reason, the set of labels is often fixed to be a two-element set, one for each branch of a choice.}

Repetitive and iterative behaviour in this language is achieved by means of procedure calls.
Calling a procedure $X$ simply invokes the choreography corresponding to $X$, given in a separate mapping of procedure definitions \cdefs.
Since choreography execution is distributed, processes do not need to synchronise when entering a procedure.
This requires a runtime term, \gencont, to denote a procedure call that only some processes have entered.
This term keeps track of both the set of processes $\pids q$ that still need to enter $X$ and the execution state of the choreography, $C$.
As we show below, the semantics of choreographies allows for out-of-order execution, and consequently some processes may start executing their part of the procedure before others have entered it.

\begin{example}[Diffie-Hellman]
\label{ex:dh-1}
Consider the Diffie-Hellman key exchange protocol~\cite{DH76} which allows two parties, $\pid p$ and $\pid q$, to establish a shared secret, $s$, that they can later use for symmetric encryption.
To implement this protocol in our choreographic language we need only communication, local computation, and a language of expressions with modular exponentiation ($b^e \mod m$)~\cite{GMP20,M23}.
The protocol assumes that participants have a private key each ($a$, $b$) and that they share a prime number $m$ and a primitive root modulo $m$, $g$.
\begin{align*}
DH = {} & \com{p}{(g^{a}\mod m)}{q}{a}; & \quad &\text{$\pid p$ computes its public key and sends it to $\pid q$}\\
& \com{q}{(g^b\mod m)}{p}{b}; && \text{$\pid q$ computes its public key and sends it to $\pid p$} \\
& \assign{p}{s}{b^a\mod m}; && \text{$\pid p$ generates the shared secret} \\
& \assign{q}{s}{a^b\mod m}; && \text{$\pid q$ generates the shared secret} \\
& \nil &&& \eoe
\end{align*}
\end{example}

\begin{example}[Zeros]
\label{ex:zeros-1}
Searching for a zero of a function is a common textbook example for program verification using Hoare-style logics \cite{ABO09}. 
In this example, we consider a version of the problem where $\pid p$ and $\pid q$ coordinate to find a zero of a function $f$ over natural numbers: $\pid p$ is responsible for selecting the values to test and $\pid q$ for evaluating $f$ and choosing whether to stop or continue searching. 
We capture this iterative protocol with the following recursive procedure.
\begin{align*}
\cdefs(Z) = {} & \com{p}{x}{q}{x};\\ 
               & \cond{q}{f(x)=0~}{~(\sel{q}{p}{\lbl{L}}; \nil)\\
               & \hspace{14ex}}{~(\sel{q}{p}{\lbl{R}}; \assign{p}{x}{1 + x}; Z)}
\end{align*}
Then, to search the domain of $f$, we run the choreography $\assign{p}{x}{0}; Z$.
\eoe
\end{example}

We define a function $\pn$ that returns the set of processes involved in an instruction or choreography.
This function is defined inductively in the natural way.
\begin{align*}
  & \pn(\gencassign) = \{\pid p\}
  && \pn(\gencom) = \pn(\gensel) = \{\pid p,\pid q\} \\
  & \pn(I;C) = \pn(I) \cup \pn(C)
  && \pn(\gencond) = \{\pid p\} \cup \pn(C_1) \cup \pn(C_2) \\
  & \pn(X) = \allpids
  && \pn(\gencont) = \pids q \cup \pn(C)
\end{align*}
For simplicity we assume that all processes are involved in all procedures; an alternative is to annotate procedure names with the set of processes they use, see~\cite{CMP21}. 
This does not affect the behaviour of any processes actually involved in the procedure, and semantically only means that a process which would otherwise be considered terminated may first have to enter some number of empty procedure calls.

\subsection{Semantics}

The semantics of choreographies uses a notion of \emph{state}, which maps each variable at each process to the value it currently stores.
It is convenient to define a \emph{local} state as a mapping from variables to values (representing the memory state at one process), and a \emph{global} state as a function $\Sigma$ such that $\Sigma(\pid p)$ is the local state at $\pid p$.

To evaluate expressions, we assume that there is an evaluation function that takes a local state as parameter, evaluates variables to their value according to the state, and proceeds homeomorphically.
In other words, evaluation maps each symbol in $\Xi$ to a function from values to values.
We assume that all choreographies and functions are well-typed, in the sense that the values stored in each variable match the types expected in the expressions in which they occur.
Furthermore, we assume that evaluation always terminates, and write $\eval e\Sigma pv$ to denote that $e$ evaluates to $v$ according to state $\Sigma(\pid p)$ (local at $\pid p$).

The formal semantics of choreographies is defined by means of a labelled transition system capturing the intuitions given above, whose rules are given in \cref{fig:sem}.
Transitions are labelled by transition labels, which abstract from the possible choreography actions that can be observed: communications of values ($\gencoml$) and labels ($\gensel$), or internal actions ($\gencondl$).
The function $\pn$ is naturally extended to these.
\begin{align*}
  \pn(\gencondl) &= \{\pid p\}
  & \pn(\gencoml) = \pn(\gensel) &= \{\pid p,\pid q\}
\end{align*}

\begin{figure}[t]
\def\rulesep{\mspace{16mu}}
\def\scalefactor{.95}
\begin{gather*}
  \scalebox{\scalefactor}{
  \infer[\rname C{Assign}]{
    \tuple{\gencassign;C,\Sigma}
    \lto{\tau@\pid p}
    \tuple{C,\Sigma[\tuple{\pid p,x}\mapsto v]}
  }{
    \eval e\Sigma pv
  }}
  \rulebreak
  \scalebox{\scalefactor}{
  \infer[\rname C{Com}]{
    \tuple{\gencom;C,\Sigma}
    \lto{\gencoml}
    \tuple{C,\Sigma[\tuple{\pid q,x}\mapsto v]}
  }{
    \eval e\Sigma pv
  }}
  \rulebreak
  \scalebox{\scalefactor}{
  \infer[\rname C{Sel}]{
    \tuple{\gensel;C,\Sigma}
    \lto{\pid p\tto \pid q[\albl]}
    \tuple{C,\Sigma}
  }{}}
  \rulesep
  \scalebox{\scalefactor}{
  \infer[\rname C{Then}]{
    \tuple{\gencond,\Sigma}
    \lto{\tau@\pid p}
    \tuple{C_1,\Sigma}
  }{
    \eval b\Sigma p\true
  }}
  \rulebreak
  \scalebox{\scalefactor}{
  \infer[\rname C{Else}]{
    \tuple{\gencond,\Sigma}
    \lto{\tau@\pid p}
    \tuple{C_2,\Sigma}
  }{
    \eval b\Sigma p\false
  }}
  \rulebreak
  \scalebox{\scalefactor}{
  \infer[\rname C{Call}]{
    \tuple{X,\Sigma}
    \lto{\tau@\pid r}
    \tuple{\cont X{\pn(C)\setminus\pid r}C,\Sigma}
  }{
    \cdefs(X)=C
  }}
  \rulebreak
  \scalebox{\scalefactor}{
  \infer[\rname C{Enter}]{
    \tuple{\gencont,\Sigma}
    \lto{\tau@\pid r}
    \tuple{\cont X{\pids q\setminus\pid r}C,\Sigma}
  }{
    \pid r \in \pids q & \pids q\setminus\pid r\neq\emptyset
  }}
  \rulesep
  \scalebox{\scalefactor}{
  \infer[\rname C{Finish}]{
    \tuple{\cont X{\pid q}C,\Sigma}
    \lto{\tau@\pid q}
    \tuple{C,\Sigma}
  }{}}
  \rulebreak
  \scalebox{\scalefactor}{
  \infer[\rname C{DelayI}]{
    \tuple{I;C,\Sigma} \lto\mu \tuple{I;C',\Sigma'}
  }{
    \tuple{C,\Sigma} \lto\mu \tuple{C',\Sigma'}
    & \pn(I)\disjoint \pn(\mu)
  }}
  \rulebreak
  \scalebox{\scalefactor}{
  \infer[\rname C{DelayC}]{
    \tuple{\gencond,\Sigma} \lto\mu \tuple{\cond pb{C_1'}{C_2'},\Sigma'}
  }{
    \tuple{C_1,\Sigma} \lto\mu \tuple{C'_1,\Sigma'}
    & \tuple{C_2,\Sigma} \lto\mu \tuple{C'_2,\Sigma'}
    & \pid p \notin \pn(\mu)
  }}
  \rulebreak
  \scalebox{\scalefactor}{
  \infer[\rname C{DelayP}]{
    \tuple{\gencont,\Sigma}
    \lto\mu
    \tuple{\cont X{\pids q}{C'},\Sigma'}
  }{
    \tuple{C,\Sigma}
    \lto\mu
    \tuple{C',\Sigma'}
    & \pids q \disjoint \pn(\mu)
  }}
\end{gather*}
\caption{Semantics}\label{fig:sem}
\end{figure}

Rules \rname C{Assign}, \rname C{Com}, \rname C{Sel}, \rname C{Then} and \rname C{Else} capture the intuition behind the different choreographic primitives given earlier.
The next three rules deal with procedure invocation: the procedure starts when one process decides to enter it, and all remaining processes are put on a ``waiting list'' (rule~\rname C{Call}); whenever a new process enters it, it is removed from the set of waiting processes (rule~\rname C{Enter}); and when the last process enters the call the set is removed (rule~\rname C{Finish}).

The last three rules deal with out-of-order execution: processes can always execute what for them is the next action, regardless of what other processes are doing.
This is modelled by rules \rname C{DelayI}, \rname C{DelayC} and \rname C{DelayP}, which allow execution of an action that is not syntactically the first instruction, conditional or procedure entering, respectively.
The side conditions in these rules state that the processes involved in the action being executed do not participate in the actions being skipped (we write $X\disjoint Y$ for $X\cap Y=\emptyset$).
Additionally, the action being performed in \rname C{DelayC} must be an action that can be made regardless of what $\pid p$ chooses.

The reflexive and transitive closure of transition is denoted by $\mlto$; we omit the sequente of transition labels, as this is immaterial for the current presentation.

For our proofs we also need the concept of \emph{head transition}, which is the transition relation defined by the first 8 rules in \cref{fig:sem} -- that is, disallowing out-of-order execution.
We write $\tuple{C,\Sigma}\hto\mu\tuple{C',\Sigma'}$ to denote that $C$ makes a head transition to $C'$, and $\mhto$ for the reflexive and transitive closure of this relation.

\section{A Hoare calculus for choreographies}\label{sec:hoare}

In this section we introduce our formal calculus for proving semantic properties of choreographies based on Hoare logic.
Our judgements are triples \hoare\varphi C\psi, interpreted as ``if choreography $C$ is executed from a state satisfying formula $\varphi$ and execution terminates, then the final state satisfies formula $\psi$''.

In this section we formally define the syntax and the semantics of this calculus, starting with the state logic -- the language in which formulas $\varphi$ and $\psi$ are written.

\subsection{State logic}

State logics in Hoare calculi typically express properties as ``variable $x$ stores a value $v$'', which are easily expressible in equational logic.
We follow this tradition, and define our state logic to be an extension of equational logic.
In order to deal with assignments, we need to be able to update formulas in a way that corresponds to the state update in rule~\rname C{Assign} -- but without computing values.
This can be achieved by substituting the expression communicated in the original formula -- but this means that expressions may suddenly refer to variables stored in different processes, so that they are no longer evaluated locally.

To deal with these issues, our state logic is parameterised on a set of expressions that is freely generated from the same signature $\Xi$, but using localised variables $\pid p.x$.
We denote these expressions as $\mathcal E$, and extend evaluation to them in the natural way.

State formulas are defined as
\[\varphi,\psi ::= (\mathcal E = \logvar X) \mid \delta \mid \varphi\wedge\varphi \mid \neg\varphi \]
where $\logvar X$ is a (logical) variable and $\delta\in\decth$, where \decth\ is a decidable theory whose terms include the logical variables.
Parameterising the language on \decth\ keeps the syntax of formulas simpler, while giving the user flexibility to define additional needed formulas.
  This is similar to our treatment of the local language.
  For example, if \decth\ includes $\logvar X>\logvar X'$, then the state logic is able to express constraints such as $\pid p.x>\pid q.y$, assuming values are integers: this can be written as $\pid p.x=\logvar X\wedge \pid q.y=\logvar Y\wedge\logvar X>\logvar Y$.
Disjunction and implication are defined as abbreviations in the usual way.

Given a state $\Sigma$, a formula $\varphi$ and an assignment $\rho$ from logical variables to values, we define $\Sigma\Vdash_\rho\varphi$, read ``$\Sigma$ satisfies $\varphi$ under $\rho$'', by the rules
\[
\infer{\Sigma\Vdash_\rho\mathcal E = \logvar X}{\mathcal E \downarrow_\Sigma \rho(\logvar X)}
\qquad
\infer{\Sigma\Vdash_\rho\delta}{\delta\in\decth & \varphi \mbox{ is true}}
\]
together with the usual rules for logical connectives.

As usual in Hoare logics, assignment is dealt with using substitution -- for example, we expect to be able to prove something like
\[
\infer{\hoare{\varphi'}{\gencassign;\nil}\varphi}{}
\]
where $\varphi'$ is obtained by $\varphi$ by substituting $\pid p.x$ with $e$.
However, simply replacing every occurrence of $\pid p.x$ with $e$ yields in general an invalid formula (due to the different variables in choreographies and state formulas).
We define the \emph{localisation of $e$ at $\pid p$}, $\loc pe$, as the (logical) expression obtained from $e$ by replacing every (choreography) variable $x$ with $\pid p.x$; and the \emph{localised substitution} $\subst{\mathcal E}qxpe$ as the expression obtained from $\mathcal E$ by replacing every occurrence of $\pid q.x$ with $\loc pe$.
(The rule for communication uses different values for $\pid p$ and $\pid q$.)
Observe that these operations can both be defined by structural recursion on expressions.
Localised substitution extends to formulas in the natural way.

\begin{example}
  Take $\varphi$ to be the formula $\pid p.x>3$ and $e$ to be the expression $y-z$.
  Replacing $\pid p.x$ with $y-z$ in $\varphi$ would yield the ill-formed formula $\pid p.(y-z)>3$.
  Instead, replacing $\pid p.x$ with $\loc p{y-z}=\pid p.y-\pid p.z$ yields the right formula $\pid p.y-\pid p.z>3$, and the above judgement becomes
  \[\infer{\hoare{\pid p.y-\pid p.z>3}{\pid p.x:=y-z;\nil}{\pid p.x>3}}{}\]
  which is syntactically well-formed.\eoe
\end{example}

We now show that an expression that has been localised to $\pid p$ is interpreted as its original evaluation in $\pid p$.
\begin{lemma}\label{lem:norm}
  Let $\Sigma$ be a state, $v$ be a value, $\logvar X$ be a logical variable and $\rho$ be an assignment such that $\rho(\logvar X)=v$.
  For any process $\pid p$ and expression $e$, $\eval e\Sigma pv$ iff $\Sigma\Vdash_\rho\loc pe=\logvar X$.
\end{lemma}
\begin{proof}
  Follows from induction on the structure of $e$.
  \qed
\end{proof}

We then show that doing a localised substitution in a formula is equivalent to changing the value of that variable in the environment.
\begin{corollary}
  \label{cor:norm}
  Let $\Sigma$ be a state, $\pid p$ be a process, $e$ be an expression and $v$ be a value such that $\eval e\Sigma pv$.
  For any formula $\varphi$ and assignment $\rho$, $\Sigma[\tuple{\pid p,x}\mapsto v]\Vdash_\rho\varphi$ iff $\Sigma\Vdash_\rho\subst\varphi qxpe$.
\end{corollary}
\begin{proof}
  By structural induction on $\varphi$.
  One of the base cases is simply \Cref{lem:norm}, while the other is trivially empty (since formulas in \decth\ are not affected by substitution).
  The two inductive cases follow directly by induction hypothesis.
  \qed
\end{proof}
 
\subsection{Hoare logic}

We are now ready to introduce the rules for our calculus, which are depicted in~\cref{fig:inf}.
To deal with procedure definitions, we need additional information about their effect on states.
This is achieved by the \emph{procedure specification map} \cprops, which maps each procedure name to a pair \tuple{\varphi,\psi} with intended meaning that the judgement \hoare{\varphi}{C}{\psi} should hold, where $C$ is the definition of $X$.
\begin{figure}[t]
  \begin{gather*}
    \infer[\rname H{Nil}]{
      \vdash_\cprops \hoare\varphi\nil\varphi
    }{}
    \qquad
    \infer[\rname H{Assign}]{
      \vdash_\cprops \hoare{\subst\varphi pxpe}{\gencassign;C}{\varphi'}
    }{
      \vdash_\cprops \hoare\varphi C{\varphi'}
    }
    \rulebreak
    \infer[\rname H{Com}]{
      \vdash_\cprops \hoare{\subst\varphi qxpe}{\gencom;C}{\varphi'}
    }{
      \vdash_\cprops \hoare\varphi C{\varphi'}
    }
    \qquad
    \infer[\rname H{Sel}]{
      \vdash_\cprops \hoare\varphi{\gensel;C}{\varphi'}
    }{
      \vdash_\cprops \hoare\varphi C{\varphi'}
    }
    \rulebreak
    \infer[\rname H{Cond}]{
      \vdash_\cprops \hoare\varphi\gencond\psi
    }{%
      \vdash_\cprops \hoare{\varphi\wedge\loc pb\logeq X\true}{C_1}\psi
      &
      \vdash_\cprops \hoare{\varphi\wedge\loc pb\logeq X\false}{C_2}\psi
      & \logvar X\mbox{ fresh}
    }
    \rulebreak
    \infer[\rname H{Call}]{
      \vdash_\cprops \hoare\varphi X\psi
    }{
      \cprops(X) = \tuple{\varphi,\psi}
    }
    \qquad
    \infer[\rname H{Call'}]{
      \vdash_\cprops \hoare\varphi\gencont\psi
    }{
      \vdash_\cprops \hoare\varphi C\psi
    }
    \rulebreak
    \infer[\rname H{Weak}]{
      \vdash_\cprops \hoare\varphi C\psi
    }{
      \decth\models \varphi\to\varphi'
      & \vdash_\cprops \hoare{\varphi'}C{\psi'}
      & \decth\models \psi'\to\psi
    }
  \end{gather*}
  \caption{Inference rules}
  \label{fig:inf}
\end{figure}

The rule for assignment \rname H{Assign} has already been motivated earlier, and is similar to the rule in standard Hoare calculi for imperative programs; likewise, rules \rname H{Nil} and \rname H{Cond} are also standard.
The notation $\loc pb\logeq X\true$ in rule \rname H{Cond} abbreviates the conjunction $\loc pb=\logvar X\wedge\logvar X=\true$.

Rule \rname H{Weak} is a weakening rule, which allows us to include reasoning in the state logic.
The notation $\decth\models\varphi$ stands for ``$\varphi$ is a valid formula''.

Rules \rname H{Com} and \rname H{Sel} adapt the intuitions behind those rules to our choreography actions --- a communication is essentially an assignment of a variable located at a different process, while selection %
does not affect the state.

Rule \rname H{Call} deals with unexpanded procedure calls by reading the corresponding judgement from the specification map, while \rname H{Call'} reflects the fact that the current state of the expanded procedure is explicitly given and a process entering a procedure does not affect the state. 

These rules only make sense if the specification map is consistent with the procedure definitions in the following sense.

\begin{definition}
  A procedure specification map \cprops\ is \emph{consistent} with a set of procedure definitions \cdefs\ if $\vdash_\cprops \hoare{\m{fst}(\cprops(X))}{\cdefs(X)}{\m{snd}(\cprops(X))}$ for every $X$, where $\m{fst}$ and $\m{snd}$ are the standard projection operators for pairs.
\end{definition}

This notion plays a similar role to the more usual concept of ``being a loop invariant'' in Hoare logics for languages with while-loops, stating that $\m{fst}(\cprops(X))$ always holds whenever $X$ is called.

\begin{example}[Diffie-Hellman, functional correctness]
\label{ex:dh-2}
Consider \cref{ex:dh-1}, and assume $\decth$ is a theory for deciding equality of arithmetic expressions with modular exponentiation.
Functional correctness for the Diffie-Hellman protocol, states if $\pid p$ and $\pid q$ have the same modulus $m$ and base $g$ then they will share the same secret $s$ once the protocol terminates.
These pre- and postconditions are captured by the following state formulas
$\varphi = (\pid p.g \logeq G \pid q.g \land \pid p.m \logeq M \pid q.m)$
and
$\psi = \pid p.s \logeq S \pid q.s$.
Thus, we can show the correctness of $DH$ by deriving $\vdash \hoare{\varphi}{DH}{\psi}$:
\[
\infer[\rname H{Weak}]
  {\vdash \hoare 
      {\varphi}
      {DH}
      {\psi}}
  {\decth\models \varphi \to\varphi_1
  &\infer[\rname H{Com}]
      {\vdash \hoare
        {\varphi_1}
        {\com{p}{(g^a\mod m)}{q}{a};\dots}
        {\psi}}
      {\infer[\rname H{Com}]
        {\vdash \hoare
            {\varphi_2}
            {\com{q}{(g^b\mod m)}{p}{b}; \dots}
            {\psi}}
        {\infer[\rname H{Assign}]
          {\vdash \hoare
              {\varphi_3}
              {\assign{p}{s}{b^a\mod m}; \dots}
              {\psi}}
            {\infer[\rname H{Assign}]
              {\vdash \hoare
                {\varphi_4}
                {\assign{q}{s}{a^b\mod m}; \nil}
                {\psi}}
              {\infer[\rname H{Nil}]
                {\vdash \hoare
                  {\psi}
                  {\nil}
                  {\psi}}
                {}}}}}}
\]
where:
\begin{align*}
\varphi_1 = {} & \subst{\varphi_2} qap{g^a\mod m} \\
          = {} & {({\pid q.g}^{\pid q.b}\mod {\pid q.m})}^{\pid p.a}\mod {\pid p.m} \logeq S ({\pid p.g}^{\pid p.a}\mod {\pid p.m})^{\pid q.b}\mod {\pid q.m} \\
\varphi_2 = {} & \subst{\varphi_3} pbq{g^b\mod m}
          = {({\pid q.b}^{\pid q.b}\mod {\pid q.m})}^{\pid p.a}\mod {\pid p.m} \logeq S {\pid q.a}^{\pid q.b}\mod {\pid q.m} \\
\varphi_3 = {} & \subst{\varphi_4} psp{b^a\mod m}
          = {\pid p.b}^{\pid p.a}\mod {\pid p.m} \logeq S {\pid q.a}^{\pid q.b}\mod {\pid q.m} \\
\varphi_4 = {} & \subst\psi qsq{a^b\mod m} = \pid p.s \logeq S {\pid q.a}^{\pid q.b}\mod {\pid q.m} &\eoe
\end{align*}
\end{example}

We can now show that this calculus is sound, in the sense that it only derives valid judgements.
Given confluence of the transition system for the semantics of choreographies~\cite{CMP21}, it suffices to show that this holds for head transitions: if execution terminates, any path of execution must lead to the same final state.

\begin{lemma}
  \label{lem:sound}
  Assume that \cprops\ is consistent with \cdefs\ and that $\vdash_\cprops \hoare\varphi C\psi$.
  For every state $\Sigma$ and assignment $\rho$, if $\Sigma\Vdash_\rho\varphi$ and $\tuple{C,\Sigma}\mhto\tuple{\nil,\Sigma'}$, then $\Sigma'\Vdash_\rho\psi$.
\end{lemma}
\begin{proof}
  The proof is by induction on the number of transitions from \tuple{C,\Sigma} to \tuple{\nil,\Sigma'}.
  Within each case, we use induction on the size of the derivation of $\vdash_\cprops \hoare\varphi C\psi$.
  We include some representative cases.

  \begin{itemize}[nosep,noitemsep]
  \item If the number of transitions is $0$, then $C=\nil$ and $\Sigma=\Sigma'$.
    The derivation of $\vdash_\cprops \hoare\varphi\nil\psi$ must then end with an application of \rname H{Nil} -- which implies that $\psi=\varphi$, establishing the thesis -- or of \rname H{Weak} -- and the induction hypothesis together with soundness of \decth\ establishes the thesis.
  \item Assume that $\tuple{C,\Sigma}\lto\gencondl\tuple{C',\Sigma'}\mlto\tuple{C'',\Sigma''}$ and that the first transition is derived by rule \rname C{Assign}.
    Then $C$ has the form $\gencassign; C'$, $\eval e\Sigma pv$, and $\Sigma'=\Sigma[\tuple{\pid p,x}\mapsto v]$.
    There are two cases, depending on the last rule applied in the derivation of $\vdash_\cprops \hoare\varphi C\psi$.

    If the derivation terminates with an application of \rname H{Assign}, then $\varphi$ is $\subst{\varphi'}pxpe$ for some formula $\varphi'$ such that $\vdash_\cprops\hoare{\varphi'}{C'}\psi$.
    By \Cref{cor:norm} it follows that $\Sigma'\Vdash_\rho\varphi'$, and the induction hypothesis applied to $C'$ establishes the thesis.

    If the derivation terminates with an application of \rname H{Weak}, then the thesis is established by the induction hypothesis over the derivation, as in the base case.
  \item Assume that $\tuple{C,\Sigma}\lto\gencondl\tuple{C',\Sigma'}\mlto\tuple{C'',\Sigma''}$ and that the first transition is derived by rule \rname C{Call}.
    Then $C$ has the form $X$, $\cont X{\pn(C)\setminus\pid r}\cdefs(X)$ and $\Sigma'=\Sigma$.
    Again there are two cases, depending on the last rule applied in the derivation of $\vdash_\cprops \hoare\varphi C\psi$.

    If the derivation terminates with an application of \rname H{Call}, then by consistency of \cprops\ and \cdefs\ we know that $\vdash_\cprops\hoare\varphi{\cdefs(X)}\psi$, from which we can infer (using \rname H{Call'}) that also $\vdash_\cprops\{\varphi\}{\cont X{\pn(C)\setminus\pid r}\cdefs(X)}\{\psi\}$.
    The induction hypothesis applies to this choreography to establish the thesis.

    If the derivation terminates with an application of \rname H{Weak}, then the thesis is established as in the previous cases.
    \qed
  \end{itemize}
\end{proof}

\begin{theorem}[Soundness]
  \label{thm:sound}
  Assume that \cprops\ is consistent with \cdefs\ and $\vdash_\cprops \hoare\varphi C\psi$.
  For every state $\Sigma$ and assignment $\rho$, if $\Sigma\Vdash_\rho\varphi$ and $\tuple{C,\Sigma}\mlto\tuple{\nil,\Sigma'}$, then $\Sigma'\Vdash_\rho\psi$.
\end{theorem}
\begin{proof}
  By the results in~\cite{CMP21}, if $\tuple{C,\Sigma}\mlto\tuple{\nil,\Sigma'}$ then also $\tuple{C,\Sigma}\mhto\tuple{\nil,\Sigma'}$ (combining deadlock-freedom with confluence).
  \Cref{lem:sound} then establishes the thesis.
  \qed
\end{proof}

\begin{example}[Zeros, functional correctness]
\label{ex:zeros-2}
Correctness for the program from \cref{ex:zeros-1} requires that if $f$ has a zero, the program terminates finding it or, equivalently, that the postcondition $\psi = ((f(\pid p.x) = 0) \logeq Z \true)$ holds.
Since there are no hypothesis on the initial state, we can use as a precondition $\phi$ any tautology (preferably one without occurrences of variables used in the program) e.g., $\varphi = (\true \logeq{T} \true)$.
The following derivation shows that the procedure specification map $\cprops(Z) = \langle \varphi, \psi\rangle$ is consistent with $\cdefs$ from \cref{ex:zeros-1}:
\[
\infer[\rname H{Weak}]
  {\vdash_\cprops \hoare
    {\varphi}
    {\cdefs(Z)}
    {\psi}}
  {\decth \models \varphi \to \varphi_1\mspace{-70mu}
  &\infer[\rname H{Com}]
    {\vdash_\cprops \hoare
      {\varphi_1}
      {\com{p}{x}{q}{x};\dots}
      {\psi}}
    {\infer[\rname H{Cond}]
      {\vdash_\cprops \hoare
        {\varphi_2}
        {\cond{q}{f(x)=0}{\dots}{\dots}}
        {\psi}}
      {\infer[\rname H{Sel}]
        {\vdash_\cprops \hoare
          {\psi}
          {\sel{q}{p}{\lbl{L}};\nil}
          {\psi}}
        {\infer[\rname H{Nil}]
          {\vdash_\cprops \hoare
            {\psi}
            {\nil}
            {\psi}}
          {}}
       &\infer[\rname H{Sel}]
         {\vdash_\cprops \hoare
           {\varphi}
           {\sel{q}{p}{\lbl{R}};\dots}
           {\psi}}
         {\infer[\rname H{Assign}]
           {\vdash_\cprops \hoare
             {\varphi}
             {\assign{p}{x}{x+1};Z}
             {\psi}}
           {\infer[\rname H{Call}]
             {\vdash_\cprops \hoare
               {\varphi}
               {Z}
               {\psi}}
             {\cprops(Z) = \langle \varphi, \psi\rangle}}}}}}
\]
where:
\begin{align*}
\varphi_1 = {} & ((f(\pid p.x) = 0) \logeq Z \true \to \psi) \land 
      ((f(\pid p.x) = 0) \logeq Z \false \to \varphi)\\
\varphi_2 = {} & ((f(\pid q.x) = 0) \logeq Z \true \to \psi) \land 
      ((f(\pid q.x) = 0) \logeq Z \false \to \varphi)
\end{align*}
The same pre- and postconditions hold for the whole program:
\[
\infer[\rname H{Assign}]
 {\vdash_\cprops \hoare
   {\varphi}
   {\assign{p}{x}{0};Z}
   {\psi}}
 {\infer[\rname H{Call}]
   {\vdash_\cprops \hoare
     {\varphi}
     {Z}
     {\psi}}
   {\cprops(Z) = \langle \varphi, \psi\rangle}}
\]
If follows from soundness, that any terminating execution ends in a state $\Sigma$ s.t., $\eval{f(x)=0}{\Sigma}{\pid p}{\true}$.
Termination follows by observing that $\pid p$ scans natural numbers starting from $0$ proceeding by single increments and thus, if $f$ has any zero, $\pid p$ will eventually send the first of them to $\pid q$ which in turn will choose to terminate the search.
\eoe
\end{example}

\section{Completeness of the Hoare calculus}
\label{sec:compl}

To establish a completeness result for our calculus, we follow standard techniques from the literature, by using a notion of \emph{weakest liberal precondition} -- the weakest assertion $\varphi$, given $\cprops$, $C$ and $\psi$, such that $\vdash_\cprops \hoare\varphi C\psi$.

\subsection{Weakest liberal preconditions}
In this section we define the weakest liberal precondition operator and show that it satisfies the expected properties.

\begin{definition}
  Let $C$ be a choreography, $\psi$ be a formula and \cprops\ be a procedure specification map.
  The \emph{weakest liberal precondition} for $C$ and $\psi$ under $\cprops$, $\wlp(C,\psi)$, is defined as follows.
  \begin{align*}
    \wlp((\gencassign;C),\psi) &= \subst{\wlp(C,\psi)}pxpe \\
    \wlp((\gencom;C),\psi) &= \subst{\wlp(C,\psi)}qxpe \\
    \wlp((\gensel;C),\psi) &= \wlp(C,\psi) \\
    \wlp(\gencond,\psi) &= (\loc pb\logeq X\true\to\wlp(C_1,\psi)) \\
    & \quad\wedge(\loc pb\logeq X\false\to\wlp(C_2,\psi)) \\
    \wlp(X,\psi) &= \m{fst}(\cprops(X)) \\
    \wlp(\gencont,\psi) &= \wlp(C,\psi) \\
    \wlp(\nil,\psi) &= \psi
  \end{align*}
\end{definition}

This operator is essentially mimicking the rules from \Cref{fig:inf}.
In the clause for conditionals, $\logvar X$ is fresh.
The only potentially surprising item is the definition of $\wlp(X,\psi)$, which ignores the actual formula $\psi$: this is again due to the fact that our results require an additional condition on \cprops\ (namely, that the conditions given are compatible with the definition of $\wlp$), which indirectly ensures that $\psi$ is also considered.

\begin{example}[Diffie-Hellman, WLP]
Consider the choreography $DH$ from \cref{ex:dh-1} and the postcondition $\psi = (\pid p.s \logeq S \pid q.s)$ from \cref{ex:dh-2}, $\wlp[](DH,\psi)$ is the formula $\varphi_1$  from \cref{ex:dh-2}.
\eoe
\end{example}

\begin{definition}
  A procedure specification map \cprops\ is \emph{adequate for $\psi$} given a set of procedure definitions \cdefs{} if, for any procedure name $X$, $\m{fst}(\cprops(X))$ is logically equivalent to $\wlp(\cdefs(X),\psi)$ and $\m{snd}(\cprops(X))=\psi$.
\end{definition}
In other words, for each $\psi$ we are interested in a mapping \cprops\ that, for each procedure, includes the right precondition that ensures that $\psi$ will hold if that procedure terminates.

\begin{example}[Zeros, WLP]
The procedure specification map $\cprops$ from \cref{ex:zeros-2} is adequate for the postcondition from the same example given the set of procedure definitions $\cdefs$ from \cref{ex:zeros-1}. 
In fact, $\wlp(\cdefs(Z),f(\pid p.x) = 0 \logeq Z \true)$ is the formula $\varphi_1$ from \cref{ex:zeros-2}, which is logically equivalent to $\m{fst}(\cprops(Z))$.
\eoe
\end{example}

The next results show that $\wlp(C,\psi)$ precisely characterises the set of states from which execution of $C$ guarantees $\psi$.
\begin{lemma}
  \label{lem:wlp1}
  Assume that \cprops\ is adequate for $\psi$ given \cdefs.
  Then, for every choreography $C$, $\vdash_\cprops \hoare{\wlp(C,\psi)}C\psi$.
\end{lemma}
\begin{proof}
  By structural induction on $C$.
  Most cases immediately follow from the definition of \wlp\ together with the induction hypothesis.
  We detail the only nontrivial ones.
  \begin{itemize}[nosep,noitemsep]
  \item If $C$ is $\gencond$, we observe that $\vdash_\cprops \hoare{\wlp(C_1,\psi)}{C_1}\psi$.
    Since $$(\wlp(\gencond,\psi)\wedge\loc pb\logeq X\true)\to\wlp(C_1,\psi)$$ is a valid propositional formula, we can apply rule \rname H{Weak} to derive $\vdash_\cprops \hoare{\wlp(\gencond,\psi)\wedge\loc pb\logeq X\true}{C_1}\psi$.
    A similar reasoning applied to $C_2$ derives the other hypothesis for rule \rname H{Cond}, and combining them establishes the thesis.
  \item If $C$ is $X$, then the thesis follows from the assumption that $\m{snd}(\cprops(X))=\psi$.
    \qed
  \end{itemize}
\end{proof}

\begin{corollary}
  \label{cor:adeq-cons}
  If \cprops\ is adequate for $\psi$ given \cdefs, then \cprops\ is consistent with \cdefs.
\end{corollary}

\begin{corollary}
  \label{cor:wlp1}
  Assume that \cprops\ is adequate for $\psi$ given \cdefs.
  For every choreography $C$, state $\Sigma$, and assignment $\rho$, if $\Sigma\Vdash_\rho\wlp(C,\psi)$ and $\tuple{C,\Sigma}\mlto\tuple{\nil,\Sigma'}$ for some state $\Sigma'$, then $\Sigma'\Vdash_\rho\psi$.
\end{corollary}
\begin{proof}
  By \Cref{lem:wlp1}, $\vdash_\cprops \hoare{\wlp(C,\psi)}C\psi$.
  By \Cref{cor:adeq-cons}, \cprops\ is consistent with \cdefs.
  The thesis then follows by \Cref{thm:sound}.
  \qed
\end{proof}

\begin{lemma}
  \label{lem:wlp2}
  Assume that \cprops\ is adequate for $\psi$ given \cdefs.
  Let $C$ be a choreography, $\Sigma$ and $\Sigma'$ be states, and $\rho$ be an assignment.
  If $\tuple{C,\Sigma}\mhto\tuple{\nil,\Sigma'}$ and $\Sigma'\Vdash_\rho\psi$, then $\Sigma\Vdash_\rho\wlp(C,\psi)$.
\end{lemma}
\begin{proof}
  By induction on the number of transitions from $C$ to $\nil$.
  If this number is $0$, then $C$ is $\nil$ and the thesis trivially follows.
  Otherwise, we detail some representative cases.
  We do case analysis on $C$ to determine the first transition.
  \begin{itemize}[nosep,noitemsep]
  \item If $C$ is $\gencassign;C''$, then $\tuple{C,\Sigma}\hto\gencondl\tuple{C'',\Sigma''}\mhto\tuple{\nil,\Sigma'}$, and $\Sigma''\Vdash_\rho\wlp(C'',\psi)$ by induction hypothesis.
    But $\Sigma''=\Sigma[\tuple{\pid p,x}\mapsto v]$ where $\eval e\Sigma pv$, hence $\Sigma\Vdash_\rho\subst{\wlp(C'',\psi)}pxpe$ by \Cref{cor:norm}, establishing the thesis.
  \item If $C$ is $\gencond$, then there are two cases.
    Assume wlog that $\eval b\Sigma p\true$.
    Then $\tuple{\gencond,\Sigma}\hto\gencondl\tuple{C_1,\Sigma}\mhto\tuple{\nil,\Sigma'}$, and $\Sigma\Vdash_\rho\wlp(C_1,\psi)$ by induction hypothesis.
    The only nontrivial case is when $\rho(\logvar X)=\true$ -- otherwise the antecedents of both implications in $\wlp(C,\psi)$ are false and the thesis trivially holds.
    If $\rho(\logvar X)=\true$, then $\Sigma\vdash_\rho\loc pb=\logvar X$ by \Cref{lem:norm}, and again both implications in $\wlp(C,\psi)$ are true (the first one has true premise and conclusion, while the premise in the second one is false).
    The case where $\eval b\Sigma p\false$ is analogous.
  \item If $C$ is $X$, then $\tuple{X,\Sigma}\mhto\tuple{\cdefs(X),\Sigma}\mhto\tuple{\nil,\Sigma'}$ by applying rules \rname C{Call}, \rname C{Enter} and \rname C{Finish} until all processes have entered $X$.
    By adequacy, $\m{fst}(\cprops(X))=\wlp(\cdefs(X),\psi)$, and the induction hypothesis establishes the thesis.
    \qed
  \end{itemize}
\end{proof}

\begin{corollary}
  \label{cor:wlp2}
  Assume that \cprops\ is adequate for $\psi$ given \cdefs.
  Let $C$ be a choreography, $\Sigma$ and $\Sigma'$ be states, and $\rho$ be an assignment.
  If $\tuple{C,\Sigma}\mlto\tuple{\nil,\Sigma'}$ and $\Sigma'\Vdash_\rho\psi$, then $\Sigma\Vdash_\rho\wlp(C,\psi)$.
\end{corollary}
\begin{proof}
  Combining \Cref{lem:wlp2} with deadlock-freedom and confluence of the semantics, as in the proof of \Cref{thm:sound}.
  \qed
\end{proof}

\subsection{Completeness}

Combining the results in the previous section, we obtain a completeness result for our calculus.

\begin{theorem}[Partial completeness]
  \label{thm:compl}
  Let $C$ be a choreography, $\varphi$ and $\psi$ be formulas, and assume that \cprops\ is adequate for $\psi$ given \cdefs.
  Assume that, for all states $\Sigma$ and $\Sigma'$ and assignment $\rho$, if $\Sigma\Vdash_\rho\varphi$ and $\tuple{C,\Sigma}\mlto\tuple{\nil,\Sigma'}$, then $\Sigma'\Vdash_\rho\psi$.
  Then $\vdash_\cprops \hoare\varphi C\psi$.
\end{theorem}
\begin{proof}
  Let $\Sigma$ be a state such that $\tuple{C,\Sigma}\mlto\tuple{\nil,\Sigma'}$, implies $\Sigma'\Vdash_\rho\psi$.
  Then $\Sigma\Vdash_\rho\wlp(C,\psi)$ by \Cref{cor:wlp2}.
  Since this is the case for all states $\Sigma$ such that $\Sigma\Vdash_\rho\varphi$, it follows that $\decth\Vdash\varphi\to\wlp(C,\psi)$.
  But $\vdash_\cprops \hoare{\wlp(C,\psi)}C\psi$ by \Cref{lem:wlp1}, whence by \rname H{Weak} the thesis holds.
  \qed
\end{proof}

\Cref{thm:sound,thm:compl} can be combined with the EPP theorem from~\cite{CMP21}, which relates the behaviour of choreographies with the behaviour of their projections, to yield results on execution of distributed implementations generated by choreographies.
This means that properties of these implementations can be analysed at the choreographic level, which is arguably simple, without the need for a specialised Hoare calculus for process languages.

\subsection{Decidability}

Finally we establish some decidability results for the Hoare calculus.
We start by pointing out that we assume $\decth$ is decidable; since propositional logic is decidable and evaluation converges, the judgments of the form $\decth\models\varphi$ that appear on the premises of rule \rname H{Weak} are also decidable.

\begin{lemma}
  \label{lem:dec1}
  The judgement $\vdash_\cprops \hoare\varphi C\psi$ is decidable.
\end{lemma}
\begin{proof}
  Assume that $\vdash_\cprops \hoare\varphi C\psi$.
  By \Cref{thm:sound}, for every state $\Sigma$ and assignment $\rho$ such that $\Sigma\Vdash_\rho\varphi$ it is the case that: if $\tuple{C,\Sigma}\mlto\tuple{\nil,\Sigma'}$, then $\Sigma'\Vdash_\rho\psi$.
  By \Cref{cor:wlp2}, this means that $\Sigma\Vdash_\rho\wlp(C,\psi)$, and therefore $\decth\models\varphi\to\wlp(C,\psi)$.

  Conversely, if $\decth\models\varphi\to\wlp(C,\psi)$, then for every state $\Sigma$ and assignment $\rho$ such that $\Sigma\Vdash_\rho\varphi$ it is the case that $\Sigma\Vdash_\rho\wlp(C,\psi)$, and therefore if $\tuple{C,\Sigma}\mlto\tuple{\nil,\Sigma'}$ it must hold that $\Sigma'\Vdash_\rho$ by \Cref{cor:wlp1}.
  By \Cref{thm:compl} this means that $\vdash_\cprops \hoare\varphi C\psi$.

  This shows that $\vdash_\cprops \hoare\varphi C \psi$ iff $\decth\models\varphi\to\wlp(C,\psi)$.
  Since $\wlp$ is computable and validity is decidable, it follows that $\vdash_\cprops \hoare\varphi C\psi$ is decidable.
  \qed
\end{proof}

Although the set of procedure names can in principle be infinite, most practical applications only use a finite subset of them.\footnote{This disallows choreographies where e.g.~each procedure $X_i$ calls procedure $X_{i+1}$, which do not occur in practice.}
In this case, consistency and adequacy also become decidable.
\begin{corollary}
  If the set of procedure names is finite, then consistency between a procedure specification map \cprops\ and a set of procedure definitions \cdefs\ is decidable.
\end{corollary}

\begin{lemma}
  If the set of procedure names is finite, then adequacy of a procedure specification map for a formula and set of procedure definitions is decidable.
\end{lemma}
\begin{proof}
  Immediate from the definition.
  \qed
\end{proof}

We end this section with a negative result: it is not possible to compute an adequate procedure specification map.
\begin{lemma}
  There is no algorithm that, given a set of procedure definitions \cdefs\ and a formula $\psi$, always returns a procedure specification map \cprops\ that is adequate for $\psi$ given \cdefs.
\end{lemma}
\begin{proof}
  Consider the formula $\psi=\bot$, which never holds.
  For any choreography $C$ and satisfiable formula $\varphi$, the judgement \hoare\varphi C\bot\ holds iff $C$ never terminates from a state that satisfies $\varphi$.

  This means that, if \cprops\ is adequate for $\bot$ given \cdefs, then $\wlp(C,\bot)$ characterises the set of states from which execution of $C$ diverges.
  In particular, $C$ never terminates if $\wlp(C,\bot)$ is logically equivalent to $\top$ -- which is decidable in our state logic.
  But Rice's Theorem implies that the class of choreographies that always diverge is undecidable, therefore \cprops\ cannot be computable.
  \qed
\end{proof}

Although this result states that adequate procedure specification maps are in general not computable, there is still the possibility that they can be shown to exist always.
Such a result would entail that our calculus is strongly complete.
We plan to investigate this issue in future work.

\section{Related Work}
\label{sec:related}
The work nearest to ours is~\cite{JV22}, where the authors propose a system for functional correctness of choreographies aimed at reasoning about distributed choices.
While they also propose a Hoare calculus for choreographies, there are some key differences wrt our work.

Firstly, they introduce a new choreographic language with significant differences from common practice in choreographic programming, e.g., they require every choice to involve every process regardless of their involvement in the branches in the condition.
By contrast, we used an existing language with standard constructs.

Secondly, the logic used in \cite{JV22} is fixed and used in the choreography language for Boolean expressions.
This coupling compromises the generality of the development, because the logic and the syntax of choreographies are not standalone.
Instead, we follow the standard two-layered approach for Hoare logic~\cite{AO19,H69}, and define a state logic that is parametric on both the language of expressions in the choreographies and the theory for reasoning about them.

As a consequence, our development is more readily applicable and adaptable to other existing choreographic languages.

The only other work combining choreographies and logic is Linear Compositional Choreographies (LCC) \cite{CMS18}, a proof theory based on linear logic for reasoning about programs that modularly combine compositional choreographies \cite{MY13} with processes.
This was inspired by previous work on the correspondence between linear propositions and session types \cite{CP10}. LCC, however, is not aimed at functional correctness: propositions represent communication behaviour rather than assertions about states.

Design-by-Contract \cite{M92} is a framework where each protocol or function is given a contract specifying its allowed input and resulting output, similar to the pre- and postconditions of Hoare logic, which has been used to reason about distributed programs from a global level.
  The first work in this line~\cite{BHTY10} defined a framework for specifying contracts for multiparty sessions.
  Being based on session types, this work more focussed on specifying properties of communicated values than ours, which lets them specify more properties than us, but also requires adding annotations to the language being reasoned about.
  An extension of this idea~\cite{MP17} describes chaperone contracts for higher-order binary sessions, which lets contracts update dynamically at runtime.
  Design-by-Contract has also been applied to microservices in the form of Whip \cite{WCD17}.
  Like our work, Whip is language-agnostic with regard to the local language, though it uses global contracts to reason directly on the local language; unlike our logic, Whip is designed for monitoring communications at runtime.

Another way of reasoning about session types is combining them with dependent types \cite{TCP11}. Like the work of \cite{BHTY10}, dependent types can be used to reason about the values being communicated, but unlike our work they are not intended to reason about pre- and postconditions.

Hoare logic has also been used to reason directly about systems of communicating processes~\cite{ANW80,LG81}.
This is far more complex than reasoning about choreographies, as it requires independently considering properties of each participant's protocol and how they are combined in the global system.

\section{Conclusions}
We have presented a novel Hoare calculus for reasoning about choreographic programs.
Our logic allows for a great deal of flexibility, since it is parametric on both the local language of the choreographic language and a decidable theory defined by the user.

We have proven that the standard properties of Hoare logics hold for our language.
Using the operational correspondence theorems for choreographies and their projections, we also showed that any properties that our logic can prove for a choreography also hold for the distributed implementation automatically generated from that choreography.

Our section on decidability left open the question of whether there always exists an adequate procedure specification map for any target formula, which we plan to investigate in future work.
We also want to look further into the issue of how our decidability results can be used to implement interesting algorithms, e.g.~for proof automation.

Our formalism only gives us guarantees for terminating execution paths, which means that we cannot infer any properties of non-terminating choreographies.
However, an inspection of the proofs of soundness and completeness (in particular, \Cref{lem:sound,lem:wlp1}) shows that these results actually guarantee something stronger, namely that the invariants described in \cprops\ must hold whenever the choreography reaches a procedure call.
We plan to use this observation as a starting point for an investigation about how our calculus can be used to assert properties of non-terminating executions of choreographies.

\paragraph{Acknowledgements.}
This work was partially supported by Villum Fonden, grant nr 29518.

\end{document}